\documentclass[reqno,11pt]{amsart}
\usepackage{amsmath,amsfonts,mathrsfs,amssymb,graphicx,color}

\usepackage[colorlinks,
            linkcolor=blue,
            anchorcolor=blue,
            citecolor=blue
            ]{hyperref}

\newcommand{\R}{{\mathbb R}}
\newcommand{\A}{{\mathcal A}}
\newcommand{\Q}{{\mathbb Q}}
 \newcommand{\N}{{\mathbb N}}
\newcommand{\dimension}{\mathrm{dim_H}}
\newcommand{\meas}{\mathrm{meas}}

 \allowdisplaybreaks[4]
\newtheorem{thm}{Theorem}[section]

\newtheorem{prop}[thm]{Proposition}%[section]
\newtheorem{rem}[thm]{Remark}%[section]
\begin{document}
%Topmatter
\title{Positive Hausdorff dimensional spectrum  for   the critical almost Mathieu operator }

\author{Bernard Helffer}
\address{Laboratoire
de Math\'ematiques Jean Leray, Universit\'{e} de Nantes and CNRS, 2 rue de la Houssini\`ere 44322 Nantes Cedex (France) and Laboratoire de Math\'ematiques d'Orsay, Univ.
Paris-Sud, Universit\'e Paris-Saclay.}
\email{Bernard.Helffer@univ-nantes.fr}

\author{Qinghui LIU}
\address{
Department of Computer Science,
Beijing Institute of Technology,
Beijing 100081, PR China.}
\email{qhliu@bit.edu.cn}

\author{Yanhui QU}
\address{Department of Mathematics, Tsinghua University, Beijing 100084, PR China.}
\email{yhqu@tsinghua.edu.cn}

\author{Qi Zhou}
\address{
Department of Mathematics, Nanjing University, Nanjing 210093, China}
\email{qizhou@nju.edu.cn}

\begin{abstract}
We show that  there exists  a dense set of   frequencies  with positive Hausdorff dimension for which   the Hausdorff dimension of the spectrum  of the critical almost Mathieu operator is positive. \end{abstract}

\maketitle

\section{Introduction}

\subsection{The context}
In this paper, we are interested in the Hausdorff dimension of the spectrum of the almost Mathieu operator  $H_{\lambda,\alpha,\theta}$ in $\ell^2(\mathbb Z)$  (denoted from now on   by AMO):
\begin{equation}\label{schro}
\ell^2(\mathbb Z) \ni u \mapsto (H_{\lambda,\alpha,\theta} u)_n= u_{n+1}+u_{n-1} +2\lambda \cos(
2\pi (n\alpha + \theta) )\, u_n\,,
\end{equation}
where $\theta\in \mathbb{T}$ is  the phase, $\alpha\in (0,1)\backslash
\Q$ is the frequency and $\lambda\in \R \setminus \{0\}$
is  the coupling constant.
The AMO was first introduced by Peierls \cite{Pe},
as a model for an electron on a 2D lattice,  submitted to  a homogeneous
magnetic field \cite{Ha,R}. This model
has been extensively studied not only because of  its importance
in  physics \cite{AOS,OA,TKNN},
but also as a fascinating mathematical object.

%Recent years, there was substential advances on the study of the spectral properties of almost Mathieu operator. Concerning the spectral measure,  as conjectured by Aubry-Andr\'e-Jitomirskaya
%\cite{AA80,Ji95,J07}, sharp phase transition appears: If $|\lambda|<1,$ then $H_{\lambda,\alpha,\theta}$ has
%purely absolutely continuous spectrum for all $\theta$ \cite{Aab}; if $1\leq |\lambda|<e^{\beta(\alpha)}$,  then $H_{\lambda,\alpha,\theta}$ has
%purely singular continuous spectrum for all $\theta$ \cite{AYZ1};  if $|\lambda|>e^{\beta(\alpha)}$, then  $H_{\lambda,\alpha,\theta}$ has
%pure point spectrum with exponentially decaying eigenfunctions
%for a.e. $\theta$ \cite{AYZ1,JL}.

The spectrum of
$H_{\lambda,\alpha,\theta}$ is a compact perfect set  in $\mathbb R$  and, observing that it is independent of $\theta$ since
$\alpha$ is irrational, we   denote it  by
$\Sigma_{\lambda,\alpha}$. For any $\lambda$, and for any irrational  $\alpha$,  $\Sigma_{\lambda,\alpha}$ is a Cantor set \cite{AJ05}. Each finite interval of  $\R\backslash\Sigma_{\lambda,\alpha}$ is called
an open gap   and we know from \cite{AJ08,AYZ2} that  if $\lambda \neq1$, then for any irrational   $\alpha$, all the spectral gaps are open, as predicted by  the  Gap Labelling Theorem \cite{JM82}.

However, little was known for the structure  of $\Sigma_{\lambda,\alpha}$ when $\lambda =1$.    $H_{1,\alpha,\theta}$ is called the critical almost Mathieu operator (or Harper's equation in physics);  it has particular importance in physics.  
{ One useful way to study $\Sigma_{\lambda,\alpha}$ is by periodic approximation \cite{AMS,k,last3},  through finer estimates of $\Sigma_{\lambda,p_n/q_n}$, where  $p_n(\alpha)/q_n(\alpha)$ is  the  best rational approximation of $\alpha,$ Last \cite{last3} shows that if $\alpha$ is not of bounded type, then $\Sigma_{1,\alpha}$ has  zero Lebesgue measure.  Finally, by the renormalization technique, Avila and Krikorian \cite{AK06} completed the proof that $\Sigma_{1,\alpha}$ has zero measure, and hence a Cantor set, for all irrational $\alpha$. Therefore, it is 
 natural to study the fractal dimensions of $\Sigma_{1,\alpha}$.}
  It was believed until the mid 1990's that
the box-counting  dimension  $\dim_B(\Sigma_{1,\alpha})$ equals to $\frac{1}{2}$ for almost every $\alpha$; one can consult
\cite{conjhalf2,conjhalf3,conjhalf1}
for numerical and heuristic arguments supporting this conjecture. But  in 1994, Wilkinson-Austin \cite{wilkinson_austin}
provided  numerical evidence that $\dim_B(\Sigma_{1,\alpha})
= 0.498$\, for \break  $\alpha = \frac{\sqrt{5}-1}{2}$
and thus conjectured that $\dim_B(\Sigma_{1,\alpha}) < \frac{1}{2}$ for every irrational $\alpha$. However, Jitomirskaya-Zhang \cite{JZ} showed that  if $\beta(\alpha) > 0$, then $\dim_B(\Sigma_{1,\alpha}) =1$, which disproved Wilkinson-Austin's conjecture.
 Here, $\beta(\alpha)$,  which  measures how Liouvillean $\alpha$,  is defined as
\begin{equation}\label{defbeta}
\beta(\alpha):=\limsup_{n\rightarrow \infty}\frac{\log
q_{n+1}(\alpha)}{q_n(\alpha)},\end{equation}
where  $q_n(\alpha)$ is the denominator of  the     $n$-th convergent of $\alpha.$ \\

We recall that the Hausdorff dimension of a set $S\subset \R$ is defined by
\[
\dimension(S) = \inf \left\{ t\in\mathbb{R^{+}} \, \big{|} \,
\lim_{\delta\rightarrow 0}\inf_{\delta
\textrm{-covers}}\sum_{n}(\meas(U_n))^t < \infty \right\},
\]
 where a $\delta$-cover of $S$ is a family  $(U_n)_n$ such that  $S \subset \cup_{n=1}^{\infty} U_n$,  and every $U_n$ is an interval of length smaller than $\delta$.
Last  \cite{last3} showed that if  $ q_{n+1}(\alpha)>q_n^4(\alpha)$ for a subsequence  of $n$, then $\dimension(\Sigma_{1,\alpha})\leq \frac{1}{2}$. We note that the set of such kind of frequencies is a
dense $G_\delta$ set which contains  $\{\alpha\in (0,1)\setminus\Q\, | \, \beta(\alpha)>0 \}$.   Shamis-Last \cite{LS}  showed that there  exists a dense set of  $\{\alpha \in (0,1)\setminus\Q | \beta(\alpha)>0 \}$, for which $\dimension(\Sigma_{1,\alpha})=0$. Recently, Avila-Shamis-Last-Zhou \cite{ALSZ} strengthened the result of \cite{LS} and showed that for any $\alpha$  in   $ \{\alpha \in (0,1)\setminus\Q \, | \, \beta(\alpha)>0 \}$,  $\dimension(\Sigma_{1,\alpha})=0$. As we can see, the  results of \cite{ALSZ,JZ,last3,LS} are more specific to Liouvillean frequency,  it is interesting to see if we can   say something about the Diophantine frequencies. \\
\subsection{Main results}
 In this paper, we will show the following:

\begin{thm}\label{thm1}
The set of frequencies
$$
\mathscr{F}:=\{\alpha\in (0,1)\setminus\Q:  \dimension(\Sigma_{1,\alpha})>0 \}\
$$
is dense in  $(0,1)\setminus \Q$ and has positive Hausdorff dimension.
\end{thm}

\begin{rem}
Combining Theorem \ref{thm1} with a result of \cite{ALSZ}, one thus know that $\mathscr{F}$ is dense in $
 \{\alpha \in (0,1)\setminus\Q \,| \, \beta(\alpha)=0 \}.
$
\end{rem}

 Indeed,  Bellissard\footnote{ Private conversation with Y. Last, circ. 1995.} conjectured that there should
exist some $\kappa \in (0, 1/2]$ such that
$\dimension(\Sigma_{1,\alpha}) = \kappa$ for almost every
$\alpha$. As far as we know, Theorem~\ref{thm1} is the \textit{first} result which shows the existence of positive Hausdorff dimension of the spectrum for the  critical almost Mathieu operator.

In fact, we can say more about the frequency and the lower bound of the Hausdorff dimension.  Denote the  continued fraction expansion of $\alpha$ as:
\[
\alpha = \frac{1}{a_1 + \frac{1}{a_2 + \frac{1}{a_3 + \ldots}}}:=[a_1,a_2,a_3,\cdots]
\]
and  define $$
A^\ast(\alpha):=\limsup_{n\to\infty} \frac{\sum_{i=1}^na_i}{n} \ \ \ \text{ and }\ \   \ G_\ast(\alpha):=\liminf_{n\to\infty} \left(\prod_{i=1}^na_i\right)^{1/n}.
$$
In 1994, Wilkinson-Austin \cite{wilkinson_austin} gave some  numerical and heuristic arguments showing that if $n$ is large enough, then
$\dim_B(\Sigma_{1,\alpha_n})$ approaches $ \frac{\log 2}{\log n}$, where
\begin{equation}\label{defalphan}
\alpha_n =[n,n,n,\cdots]\,.
\end{equation}
Motivated by this paper, we will show the following:

\begin{thm}\label{thm2}
Fix $M >0 $ and $ \widehat m \in \mathbb N$ s.t. $\widehat m \ge2$. There exist constants $C>1$  and $C' >0$
such that for any $\alpha=[a_1,a_2,a_3,\cdots] $ with, for some $m\in\{0,\cdots,\widehat m\}$,
$$
1\leq a_i\le M, \ (1\le i\le m);\ \ \ a_i\ge C, \ (i>m)
$$
we have
\begin{equation}\label{Sigma-alpha}
\dimension(\Sigma_{1,\alpha}) \geq \frac{1}{C'}   \frac{\log G_\ast(\alpha)}{A^\ast(\alpha)}.
\end{equation}
Consequently, if $n\geq C$, then,  for $\alpha_n$ defined by \eqref{defalphan},
\begin{equation}\label{Sigma-alpha-n}
\dimension(\Sigma_{1,\alpha_n}) \geq \frac{1}{C'}  \, \frac{\log n}{n}\,.
\end{equation}
\end{thm}

\begin{rem}
Now considering Theorem \ref{thm2} and the results of \cite{ALSZ,LS}, it is reasonable to conjecture that if $A^\ast(\alpha)=\infty$, then $\dimension(\Sigma_{1,\alpha}) =0\,$.
\end{rem}

\section{The covering structure of the spectrum: main statement}

Assume
$
\alpha=[a_1,a_2,a_3,\cdots].
$ Helffer and  Sj\"ostrand \cite{HS,HS2} exhibited a fine covering structure of the spectrum for a special class of frequencies. 
In the following,  we describe  a reformulation of this statement using the coding language just for the convience of the proof, readers can consult the appendix for its original formulation. 
Before to explain the general construction we detail the two first steps.\\

  {\bf In the $0$-th step},    there are $q_{m}(\alpha)$ disjoint bands. We code them by
$$
J_1,J_2,\cdots, J_{q_{m}(\alpha)}.
$$
 Thus
 $$\Theta_0=\{1,2,\cdots,q_{m}(\alpha)\}$$
 is the set of words to code the bands of $0$-generation.\\

{\bf In the first step}, for each band $J_i$, there are finitely many sub-bands inside it. For the mid-band, we code it by  $J_{i \cdot 0}$. Here $\cdot$ means the concatenation of words.   On the left, there are $m_i$ sub-bands with good estimates on the band-length, which we code by $J_{i\cdot (-m_i) },\cdots J_{i\cdot (-1)}$. Similarly, on the right, there are $n_i$ sub-bands, which we code by $J_{i\cdot 1},\cdots, J_{i\cdot n_i}$. Write
$$\mathcal{A}_i=\{-m_i,\cdots, -1, 1,\cdots,n_i\}$$
 (here $\mathcal{A}$ means alphabet), then
$\{i\cdot j: j\in \mathcal{A}_i \}$ is the set of words to code the subbands inside $J_i$. Now if we collect the words for all $i\in \Theta_0$, we obtain $\Theta_1:$
$$
\Theta_1=\{i\cdot j: i\in \Theta_0, j\in \mathcal{A}_i\}.
$$
 Since for each $i\in \Theta_0$, there is only one  mid-band $J_{i\cdot 0}$ inside $J_i$,
 $$\Omega_1=\{i\cdot 0: i\in \Theta_0\}$$
  is the set of words to code the mid bands of $1$-generation.

We now define inductively two sequences of words  $\Theta_k, \Omega_k$ as follows.
%Given a positive integer vector $(m_\epsilon, n_\epsilon)$, write
%$$
%\A_\epsilon:=\{-m_\epsilon,\cdots, -1,1, \cdots, n_\epsilon\}.
%$$
%Define
%$$
%\Sigma_1:=\A_\epsilon \ \text{ and }\ \ \Omega_1=\{0\}.
%$$
Assume $\Theta_k$ and $\Omega_k$ have been defined for some  $k\ge 1\,$. For any $\theta\in \Theta_k$, fix a  vector $(m_\theta, n_\theta)\in \mathbb N^* \times \mathbb N^*$  (where {$\mathbb N^*=\mathbb N \setminus \{0\}$})
%\footnote{{\clm need to unify the notation.} {\clr Bernard:  I always take the "Bourbaki" convention for $\mathbb N$. That is $0$ belongs to $\mathbb N$.}}
and write
$$
\A_\theta:=\{-m_\theta,\cdots, -1,1, \cdots, n_\theta\}.
$$
Define
\begin{eqnarray*}
\Theta_{k+1}&:=&\{\theta\cdot i: \ \theta\in \Theta_k;\ i\in \A_\theta\}\\
\Omega_{k+1}&:=& \Omega_k\cup \{\theta\cdot 0: \theta\in \Theta_k\}.
\end{eqnarray*}

Write $\Omega_0=\emptyset\,.$ Define
$$
\Theta:=\bigcup_{k\ge 0}\Theta_k\ \  \mbox{ and}\ \   \Omega:=\bigcup_{k\ge 0}\Omega_k.
$$

As continuation of \cite{HS}, the following theorem is proved in \cite{HS2}, the readers can just consult the appendix for its original statement.

\begin{thm} \label{hs}
 Fix  $ \widehat m \in \N\,,$ and $M\ge2$. Then there exist $\epsilon_1>0$  and, for $0 < \epsilon_0\leq \epsilon_1$, some  constants
 $ C_1>0, b_2 > b_1>0\,, c_1>0, d_2>d_1>0$ such that  if $\alpha=[a_1,a_2,a_3,\cdots]$ and for some $0\le m \le  \widehat m$
\begin{eqnarray}\label{introm}
 \left\{ \begin{array}{ccc} 1\leq   a_\ell \le M\,, & \ell \leq m \\
a_\ell \ge C_1\,, & \ell \geq m+1\end{array} \right.\,,
\end{eqnarray} then there exists a  sequence $\{(m_\theta, n_\theta):\theta\in \Theta\}$ with
\begin{equation}\label{m-n-theta} 
b_1 \,a_{k+m}\le m_\theta \leq b_2 \,a_{k+m} \,\mbox{ and }  \,b_1 \,a_{k+m}\le n_\theta \le b_2 \,a_{k+m}\,,\, \forall k\ge1,\, \forall \theta\in \Theta_{k-1}\,,
\end{equation}
and  a family of bands
$$
\{J_\theta: \theta\in \Omega\cup \Theta\}
$$
such that:
\begin{itemize}
\item[(i)]
For each $k\ge 0$, $\{J_\theta: \theta\in \Omega_{k}\cup \Theta_k\}$ is a covering of  $\Sigma_{1,\alpha}$:
$$
\Sigma_{1,\alpha}\subset \bigcup_{\theta\in \Omega_{k}\cup \Theta_k} J_\theta\,.
$$
\item[(ii)]   For each $k\ge1$ and $\theta\in \Theta_{k-1}\,$,
$$
\partial J_\theta\subset \Sigma_{1,\alpha}.
$$
\end{itemize}
 For each    $i\in \A_\theta\cup\{0\},$
$$
J_{\theta \cdot i}\subset J_\theta,
$$
 $J_{\theta \cdot (i+1)}$ is on the right of $J_{\theta\cdot i}\,$. Moreover,
 \begin{equation}\label{low}
\frac{c_1}{a_{k+m}}\le  \frac{d(J_{\theta\cdot (i+1)},J_{\theta\cdot i})}{|J_\theta|}.
 \end{equation}

 (iii) For each $k\ge 1$ and $\theta\in \Theta_{k-1}$,
 $$
 \frac{|J_{\theta\cdot 0}|}{|J_{\theta}|}\le \epsilon_0\,;\ \ e^{-d_2 a_{k+m}}\le   \frac{|J_{\theta\cdot i}|}{|J_{\theta}|}\le e^{-d_1 a_{k+m}}, \ (i\in \A_\theta)\,.
 $$
 \end{thm}
% \begin{rem} \clm 
% It could be useful to analyze the behavior of  optimal  $b_1(\epsilon_0,\hat m,\hat M)$ and $b_2(\epsilon_0,\hat m, \hat M) $   as $\epsilon_0 \rightarrow 0$.
% This  results of a Weyl formula which shows that they can be chosen as  very close to $1$ as $\epsilon_0 \rightarrow 0\,$. This is coherent, with the analysis of the rational case. Similarly, one can show that  $d_1(\epsilon_0)$ tends to $0$ as $\epsilon_0 \rightarrow 0$ and  that $d_2(\epsilon_0)$ can be chosen such that  it tends to a positive limit  (see \cite{HS}). 
% \end{rem}

%{\clr Yanhui: I have a question to Bernard: is it indeed possible that $b_1$ can be very large?
%By my experience on AMO, to study the spectrum, another efficient way is by periodic approximation.  That is take $\alpha_n=p_n/q_n$ and study the spectrum of $\Sigma_{1,\alpha_n,0}$. Roughly in step-n, there are $q_n$ bands, and from step-n to step-(n+1), each band will split to around $a_{n+1}$ subbands.  From this intuition, $b_1$ should be around $1$. Of course, in your work with Sjostrand, you used a totally different approach,  so maybe it is quite different from the periodic approximation.
%}
%{\clr Bernard: You are completely right. Semi-classical analysis is compatible with what one gets in the rational case.}

 %%%%%%%%%%%%%%%%%%%%%%%%%%%%%%%
 \section{Harper's model and semi-classical analysis}
 \subsection{Harper's model  in the rational case}
We refer also to the survey of J.~Bellissard \cite{bel} for a state of the art in 1991. When $\alpha$ is irrational, an equivalent way (observing that in this case
 $\Sigma_{\lambda,\alpha} = \cup_\theta \Sigma_{\lambda,\alpha,\theta}$)  for the analysis of the spectrum in the case of a square lattice is  to consider  the so-called Harper model, which this time is defined on $\ell^2(\mathbb Z^2,\mathbb C)$ by
   $$
  ( H^\gamma  u)_{m,n} :=  u_{m+1,n} + u_{m-1,n} +  e^{i  \gamma m} u_{m,n+1} +  e^{-i \gamma m} u_{m,n-1}\,,
  $$
  where $\gamma$ denotes the flux of the constant magnetic field through the fundamental cell of the lattice.

  When $\alpha :=\frac{\gamma}{2 \pi}$ is a rational,   Floquet theory permits to show that the spectrum is the union   of the spectra of a family of  $q\times q$ matrices $M_{p,q}(\theta_1,\theta_2)$ depending
  on the quasi-momenta $\theta=(\theta_1,\theta_2) \in \mathbb R^2$. In this way, we get that the spectrum is the union of $q$ bands
  $[\gamma_\ell,\delta_\ell]$  with $\delta_\ell \leq \gamma_{\ell +1}$, with strict inequality, except when $q$ is even for $\ell =\frac q 2$.   More precisely, when
$\gamma=2\pi p/q$,  where $p\in\mathbb Z$ and $q\in\mathbb N^*$ are relatively prime,
the two following matrices {in $M_q(\mathbb C)$ play an important role:
\begin{equation*}
J_{p,q}={\rm diag}(e^{i(j-1)\gamma})\,,
\end{equation*}
and
\begin{equation*}
(K_q)_{jk}= 1\; \mbox{ if } k\equiv j+1\, [q]\,,\,0 \mbox{ else.}\,
\end{equation*}
In the case of Harper, the family of matrices is
\begin{equation}\label{bs1}
M_{p,q}(\theta_1,\theta_2) =
  e^{i\theta_1} J_{p,q} + e^{-i\theta_1} J_{p,q}^* + e^{i\theta_2} K_q + e^{-i \theta_2} K_q^* \,.
\end{equation}
The Chambers formula gives a very elegant formula for this determinant:
\begin{equation}\label{chambers}
\det (M_{p,q}(\theta_1,\theta_2)-\lambda)  = f_{p,q}(\lambda)+ (-1)^{q+1} 2 \left(\cos q\theta_1 + \cos q \theta_2\right)\,,
\end{equation}
where $ f_{p,q}$ is a polynomial of degree $q$.
Each band  $I_\ell$ is described by a solution $\lambda_\ell (\theta_1,\theta_2)$ of the Chambers equation which can be expressed in the form
\begin{equation}\label{chambers2}
\lambda_\ell (\theta_1,\theta_2)=\varphi_{\ell,p,q} ( 2 \left(\cos q\theta_1 + \cos q \theta_2\right))\,.
\end{equation}
\subsection{Semi-classical analysis}\label{s7}
It can be shown that the spectrum $\Sigma_{\lambda,\alpha}$ is the same as the spectrum of the  operator $\widehat H_\gamma$  acting on $L^2(\mathbb R)$ defined by
$$
L^2(\mathbb R) \ni u \mapsto  \lambda  (\tau_\gamma  + \tau_{- \gamma}) u + 2 \cos x \, u \in L^2(\mathbb R)\,,
$$
where $\tau_\gamma$ is defined by $\tau_\gamma u (x)= u(x-\gamma)$.

In other words, $\widehat H_\gamma$ is the  Weyl's $\gamma$-quantization of the symbol $(x,\xi) \mapsto  2(\lambda \cos \xi + \cos x) $.
Let us recall what is meant by this.  The   symbols are $C^\infty$ functions of $\mathbb R^2\ni (x,\xi)\mapsto p(x,\xi,h)$ which are  $C^\infty$ bounded  (in our case they are in addition $(2\pi)$-periodic in each variable)  depending on a semi-classical
parameter $\gamma = h\in[-h_0,0)\cup (0,h_0]$, $h_0>0$ (view as ``little'') and satisfying
\begin{equation}
 \forall (j,k)\in\N^2\,;\,\exists C_{j,k}\,;\,\forall (x,\xi)\in\R^2,\,
|\partial_x^j\partial_\xi^k p(x,\xi,h)|\leq C_{j,k} \,.\end{equation}

The  Weyl quantization of the symbol $p$ (for $h\neq 0$, $|h|\leq h_0$) is the pseudodifferential operator acting on $L^2(\R)$ by
\begin{equation}
\text{Op}^W_h(p)u(x)= \frac{1}{2\pi h}\int\!\!\!\int e^{i(x-y)\xi/h} p(\frac{x+y}{2},\xi,h)\, u(y)\,dy\,d\xi \,.
\end{equation}
An important fact is that when $p$ is real valued the associate operator is self-adjoint. This approach is only powerful when $h$ is small. This is what is call semi-classical analysis. Hence, this leads as to analyze the spectrum near $\alpha=0$.
\subsection{Semi-classical analysis near $\alpha=0$ when $\lambda=1$.}
Here we have to analyze the spectrum of the $h$-pseudodifferential operator of symbol $p(x,\xi):=2( \cos x + \cos \xi)$, i.e.
$ 2( \cos x + \cos h D_x) = 2 \cos x +  (\tau_h + \tau_{-h})$. For $ E \in [-4,4]\setminus \{0\}$, semi-classical analysis says that one has first
 to look at the energy level  $p(x,\xi)=E$. Here we see that the energy level is the union of curves which can be indexed by $\mathbb Z^2$. Modulo $\mathcal O (h^\infty)$, the spectrum is obtained by looking at a Hamiltonian has a symbol $\tilde p (x,\xi)$ such that $\tilde p (x,\xi) \rightarrow + \infty$ and
  $\tilde p^{-1} (E-\epsilon, E +\epsilon) = p^{-1} (E-\epsilon, E +\epsilon) \cap (-\pi,\pi)^2$. For this operator the spectrum is discrete and is given near $E$ by the Bohr-Sommerfeld quantization which determines a sequence of eigenvalues $\lambda_k(h)$ in $(-4, -\epsilon_0) \cup (\epsilon_0,4)$, whose asymptotic is known modulo $\mathcal O (h^\infty)$. Coming back to the initial problem, the analysis of the tunneling between the different wells leads to a localization of the
   spectrum in a family of intervals whose center is $\mathcal O (h^\infty)$ close to the sequence $\lambda_k (h)$ and whose size is exponentially small and can be measured precisely.

   The next step is that in each of these intervals the spectral analysis of the  restriction of our initial operator appears to be the $h'$-quantization
    of an operator whose symbol is close (after renormalization) to $\cos \xi + \cos x$ and where $\frac{h'}{2\pi} \equiv \frac{2\pi}{h}$ (modulo $\mathbb Z$).
    This relation gives a strong link between the corresponding continued fraction expansions of $\alpha =\frac{h}{2\pi}$ and $\alpha'=\frac{h'}{2\pi}$ since
$
\frac{h'}{2\pi} =[a_2,a_3,\cdots]\,.
$
    Under our assumption $h'$ is small (because  $a_2$ is assumed to be large) and we can redo the same analysis leading again to a new familly of intervals. Of course, this is quite technical to control the uniformity of the constants appearing in the renormalization procedure.
    In the first step $E=0$ corresponds to  critical values of $p(x,\xi)$  with a saddle point. The analysis near this point is much more involved and the introduction of this $\epsilon_0$ permits to avoid this analysis  in our paper like in \cite{HS2}.

\begin{figure}
  \centering
  \includegraphics[width=8cm]{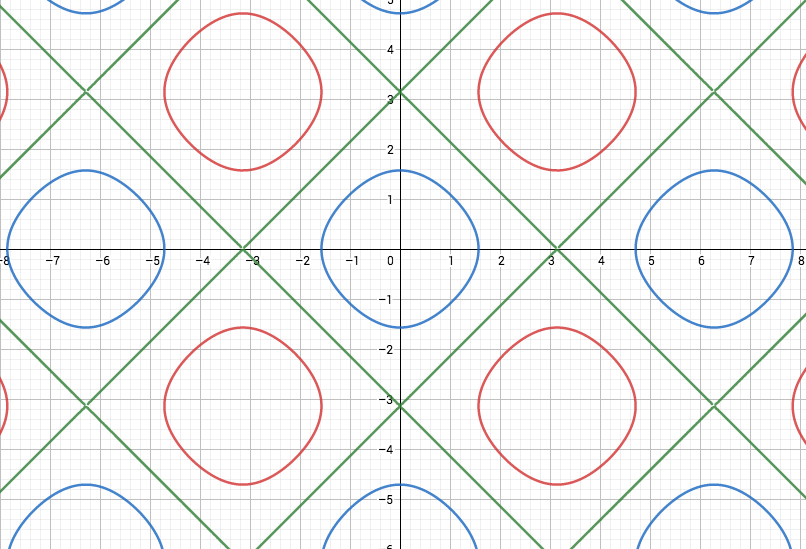}\\
  \caption{Energy levels for $(x,\xi) \mapsto \cos x + \cos \xi$.  In Green $E=0$, in blue $E=1$, in red $E=-1$. }\label{F-R1}
\end{figure}
\subsection{Semi-classical analysis near a rational}
When $\alpha = \frac p q +\hbar $, one can show that $\Sigma_{\alpha,1}$ is the spectrum of the Weyl $h$-quantization of
$M_{p,q} (x,\xi)$ (with $h= 2\pi \hbar$). Hence we have again to perform a semi-classical analysis but this time for a system of $h$-pseudodifferential operator.
As shown by Chamber's formula \eqref{chambers}, the semi-classical analysis for $\hbar$ small is strongly related with the semi-classical analysis of a function of  $\cos q x + \cos q hD_x$ (see \eqref{chambers2}). Hence, near each of the $q$ bands, we have to perform a spectral analysis which is close to the analysis of the previous subsection. When $q$ is even, there is a need
 for a special analysis for the two central touching bands (see Appendix).

%%%%%%%%%%%%%%%%%%%%%
\section{The lower bound of the Hausdorff dimension of the spectrum}

In this section, we prove Theorem \ref{thm1} and Theorem \ref{thm2}.  At first we construct a Cantor subset of the spectrum based on Theorem \ref{hs}. Then we estimate the lower bound of the Hausdroff dimension of this Cantor set. Finally, we prove the main theorems.

\subsection{A subset of the spectrum}\

Fix a frequency $\alpha=[a_1,a_2,\cdots]$ satisfying the condition in Theorem \ref{hs}.
 For each $n\ge1$, we introduce
$$
\kappa_n:=[b_1a_{m+n}]\,,
$$
where  $[x]$  denotes  the integer part of the number $x$,  $b_1$ is defined in Theorem~\ref{hs}.

In the following, we will construct a Cantor subset of $\Sigma_{1,\alpha}\cap J_1$, where $J_1$ the most left interval of 0-generation.
For all $n\ge1$, define
$$
\widehat \Theta_{n}:=\{1\}\times\prod_{i=1}^n\{1,2,\cdots, \kappa_i\}\ \ \ \text{ and }\ \ \ \mathscr{C}_n:=\bigcup_{\theta\in \widehat \Theta_n}J_\theta.
$$
By \eqref{m-n-theta}, we have $\widehat \Theta_n\subset \Theta_n$. Define the set of  codings as
$$
 \widehat \Theta_\infty:=\{1\}\times\prod_{i=1}^\infty\{1,2,\cdots, \kappa_i\}.
$$
For any $\omega=\omega_0\omega_1\omega_2\cdots\in \widehat \Theta_\infty$, we denote the prefix of length $n+1$ of $\omega$ by
$$
\omega|_n:=\omega_0\cdots\omega_n.
$$

By  Theorem~\ref{hs}, $\mathscr{C}_n$ is a decreasing sequence of compact sets. Define
$$
\mathscr{C}:=\bigcap_{n\ge0}\mathscr{C}_n.
$$
 By Theorem \ref{hs} (iii) and \eqref{introm}, for any $\theta\in \widehat\Theta_n$, we have
$$
|J_\theta|=|J_1|\prod_{k=1}^n\frac{|J_{\theta|_k}|}{|J_{\theta|_{k-1}}|}\le |J_1|\prod_{k=1}^n e^{-d_1a_{k+m}}\le |J_1|e^{-d_1C_1n}.
$$
Thus, $|J_\theta|$   tends to $0$ when the length of the word $\theta$ tends to infinity. So,  $\mathscr{C}$ is a Cantor set. Indeed, the definition of $\widehat \Theta_n$ make sure that the chosen band always avoid the excluded intervals of size $\approx \epsilon_0$ at each step.

\begin{prop}\label{prop-3.1}
$\mathscr{C}\subset \Sigma_{1,\alpha}\cap J_1$.
\end{prop}

\begin{proof}
By Theorem \ref{hs} (ii), for each $\theta\in \widehat \Theta_n$ we have $\partial J_\theta\subset \Sigma_{1,\alpha}\cap J_1$. By the construction of $\mathscr{C}$ and Theorem \ref{hs}, for any $x\in\mathscr{C}$, there exists a unique coding sequence $\omega\in\widehat \Theta_\infty$ such that
$$
\{x\}=\bigcap_{n\ge0}J_{\omega|_n}.
$$
Let $x_n$ be the left endpoint of $J_{\omega|_n}$, then $x_n\in\Sigma_{1,\alpha}\cap J_1$ and $x_n\to x.$
 Since $\Sigma_{1,\alpha}\cap J_1$ is compact, we conclude that $x\in \Sigma_{1,\alpha}\cap J_1.$
\end{proof}

\subsection{The lower bound of $\dim_H \mathscr{C}$ }\

 We need to use a basic fact from fractal geometry which we recall now.
Assume  $X\subset \R$ is a Borel set. Let $\mu$ be a probability measure supported on $X$. For any $x\in X$, the  lower local dimension of $\mu$ at $x$ is defined by
$$
\underline{d}_\mu(x):=\liminf_{r\to 0}\frac{\log\mu( B(x,r))}{\log r}.
$$

\begin{prop} [Prop. 10.1 of \cite{Fal}] \label{fal} 
If $\underline{d}_\mu(x)\ge d_0$ for $\mu$-a.e. $x\in X$, then $\dim_H X\ge d_0\,.$
\end{prop}

 For given $c_1>0$ and $d_1>0$, we introduce $\tilde C (c_1,d_1) >1/d_1$  such that
\begin{equation}\label{tilde-C}
\frac{c_1\exp(d_1\tilde C)}{\tilde C}>2\,.
\end{equation}

\begin{prop}\label{prop-3.2}
   Let  $ \widehat m$, $M$,  $\epsilon_0 \in (0,\epsilon_1),$ $b_1,$ $c_1$, $d_1$, $d_
2$ be the constants given as in Theorem \ref{hs}.  Fix $m$ with $0\le m\le  \widehat m.$ Assume that  $\alpha=[a_1,a_2,\cdots]$ satisfies   $1\leq a_i\le M$ for $1\le i\le m$ and  $ a_i\ge \max\{C_1,\tilde C (c_1,d_1) \}$ for $i\ge m+1$. Then
$$
\dim_H \mathscr{C}\ge \frac{\log (b_1/2)+\log G_\ast(\alpha)}{d_2\, A^\ast(\alpha)}.
$$
\end{prop}

\begin{proof}
 For any $n\ge 1$, if we write
 $$
 \mathcal{F}_n:=\{\mathscr{C}\cap J_\theta: \theta\in \hat\Theta_n\}\,,
 $$
 then $\mathcal{F}_n$ is a finite $\sigma$-algebra on $\mathscr{C}$. We define a measure $\mu_n$ on $(\mathscr{C}, \mathcal{F}_n)$ as
 $$
 \mu_n(\mathscr{C}\cap J_\theta):=\frac{1}{\prod_{i=1}^n\kappa_i}.
 $$
 By  Carath\'eodory extension theorem (see for example \cite{Ro}), there exists a unique probability measure $\mu$ on $\mathscr{C}$ such that $\mu|_{\mathcal{F}_n}=\mu_n.$

Now we fix  $x\in \mathscr{C}$, and will estimate $\underline{d}_\mu(x)$. Let $\omega\in\widehat \Theta_\infty$  be the coding of $x$, that is,
$$
\{x\}=\bigcap_{n\ge 1}J_{\omega|_n}.
$$
Take $r>0$. There exists a unique $n_r\in \N$ such that
$$
J_{\omega|_{n_r}}\subset B(x,r),\ \ \  J_{\omega|_{n_r-1}}\not\subset B(x,r).
$$
This implies that
\begin{equation}\label{r-esti}
|J_{\omega|_{n_r}}|< 2r, \ \ \ |J_{\omega|_{n_r-1}}|\ge r.
\end{equation}

%%%%%%%%%%%%%

\begin{figure}[h]
\begin{center}
\includegraphics[width=1\textwidth]{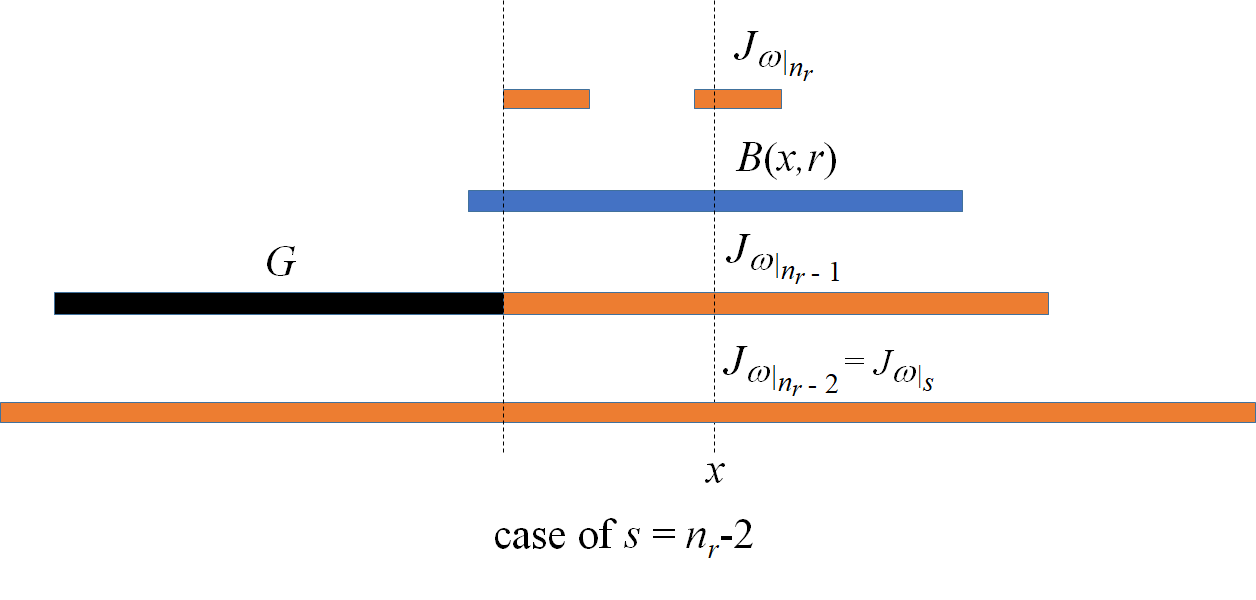}
\caption{  }\label{figure-2}
\end{center}
\end{figure}

%%%%%%%%%

~\\
\noindent {\bf Claim:} $B(x,r)$ can only intersect one band of $(n_r-1)$-generation, which is $J_{\omega|_{n_r-1}}$.
\\

\noindent $\lhd$
 If $B(x,r)\subset  J_{\omega|_{n_r-1}}$, then the claim holds trivially. In the following, we assume $B(x,r)\not\subset  J_{\omega|_{n_r-1}}$.  
To prove the claim, we only need to show that  the lengths of the gaps in the left and right of $J_{\omega|_{n_r-1}}$ are all bigger than $2r$.

We  call any component of
${\rm Co}(\mathscr{C})\setminus \mathscr{C}$ a gap of $\mathscr{C}$, where ${\rm Co}(\mathscr{C})$ is the convex hull of $\mathscr{C}$. A gap $G$ is called of order $k$, if $G$ is a subset of some  band of $k$-generation but is not a subset of any  band of $(k+1)$-generation.  

%{\clr
%Since $J_{\omega|_{n_r-1}}\not\subset B(x,r)$,
%it is  impossible that both of the gap in the left side and the gap in the right side of $J_{\omega|_{n_r-1}}$
%intersect with $B(x,r)$. Without loose of generality, we suppose that only the gap in the left side of
%$J_{\omega|_{n_r-1}}$ intersect with $B(x,r)$.} 
Denote the  gap in the left of  $J_{\omega|_{n_r-1}}$ by $G$. By our construction of $\mathscr{C}$,  G is a gap of order $s$ for some $ s \le n_r-2$ (see Figure \ref{figure-2} and Figure \ref{figure-3}). That is, $G\subset J_{\theta}$ for some $\theta\in \hat\Theta_s$,  but $G$ is not a subset of any band of  $(s+1)$-generation.

We claim that $\theta=\omega|_{s}$. If otherwise, there exists some $\tilde s\le s$ such that $\hat\theta:=\theta_0\cdots \theta_{\tilde s-1}=\omega|_{\tilde s-1}$  but $\theta_{\tilde s}\ne \omega_{\tilde s}.$ Then $J_\theta$ and $J_{\omega|_s}$ are two different descendants of same generation of $J_{\hat\theta}$, hence disjoint. In particular, $\bar{G}\cap J_{\omega|_{s}}=\emptyset$, since $\bar{G}\subset J_\theta$, where $\bar{G}$ is the closure of $G.$  On the other hand, the right endpoint of $G$ is the left endpoint of $J_{\omega|_{n_r-1}}$, thus $\bar{G}\cap J_{\omega|_{n_r-1}}\ne\emptyset\,$. Notice that, $J_{\omega|_{n_r-1}} \subseteq J_{\omega|_{s}}$ since $s\le n_r-1$. So, $\bar{G}\cap J_{\omega|_{s}}\ne\emptyset$, which is a contradiction.

Thus $G\subset J_{\omega|_s}$ is a gap of order $s$.
Notice that, $J_{\omega|_{n_r-1}}\subset J_{\omega|_{s+1}}$ since $s+1\le n_r-1$. Also, since $a_{m+s+1}\ge \tilde C$, by Theorem \ref{hs} (ii), (iii),  \eqref{tilde-C} and \eqref{r-esti}, we have
\begin{equation}\label{only-for-exp}
|G|\ge \frac{c_1|J_{\omega|_{s}}|}{a_{m+s+1}}\ge \frac{c_1e^{d_1a_{m+s+1}}|J_{\omega|_{s+1}}|}{a_{m+s+1}}\ge 2|J_{\omega|_{s+1}}|\ge 2|J_{\omega|_{n_r-1}}|\ge 2r.
\end{equation}
By the same argument, we can show that the gap in the right of $J_{\omega|_{n_r-1}}$ also has length bigger than $2r$. Then the claim follows.
\hfill $\rhd$

\begin{figure}[h]
\begin{center}
\includegraphics[width=1\textwidth]{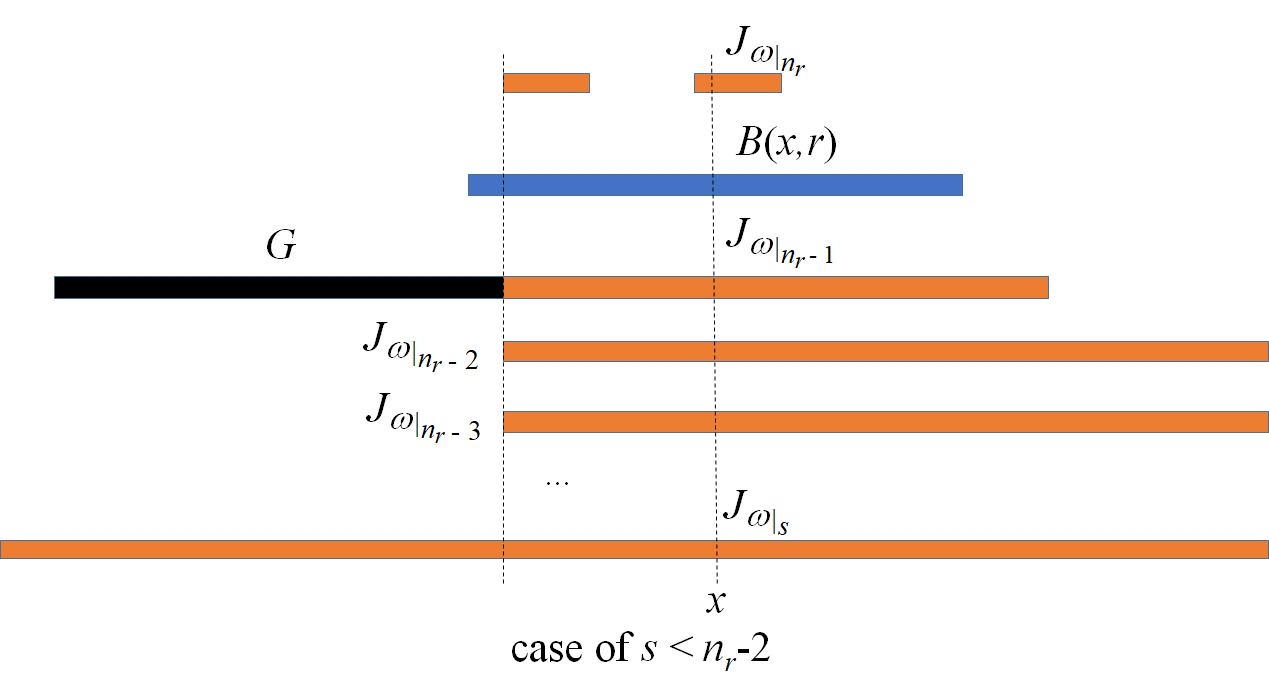}
\caption{  }\label{figure-3}
\end{center}
\end{figure}

By the claim,  we have
\begin{equation}\label{mu-B}
\mu(B(x,r))\le \mu(J_{\omega|_{n_r-1}})=\frac{1}{\prod_{i=1}^{n_r-1} \kappa_i}.
\end{equation}

By Theorem \ref{hs} (iii) and the assumption, we have
$$
  \frac{|J_{\omega|_k}|}{|J_{\omega|_{k-1}}|}\ge e^{-d_2 a_{m+k}}.
$$
Consequently,
\begin{eqnarray*}
  |J_{\omega|_{n_r}}|=|J_1|\prod_{k=1}^{n_r}\frac{|J_{\omega|_k}|}{|J_{\omega|_{k-1}}|}\ge |J_1| e^{-d_2 \sum_{k=1}^{n_r}a_{m+k}}.
\end{eqnarray*}
By \eqref{r-esti}, we get
\begin{equation}\label{log-r}
\log r\ge\log|J_1|-\log2 -d_2\sum_{k=1}^{n_r}a_{m+k}\,.
\end{equation}
Recall  that $\kappa_i=[b_1a_{m+i}]$ and by the assumption, $a_{m+i}\ge 2/b_1$. So we have  $\kappa_i \ge b_1a_{m+i}/2\,.$ 
   Combining \eqref{mu-B} and \eqref{log-r}, we get
\begin{eqnarray*}
\frac{\log\mu( B(x,r))}{\log r}&\ge& \frac{\sum_{i=1}^{n_r-1}\log \kappa_i}{d_2\sum_{k=1}^{n_r}a_{m+k}+\log2-\log|J_1|}\\
&\ge&\frac{(n_r-1)\log \frac{b_1}{2}+\sum_{i=1}^{n_r-1}\log a_{m+i}}{d_2\sum_{k=1}^{n_r}a_{m+k}+\log2-\log|J_1|}.
\end{eqnarray*}
By taking the lower limit, we get
$$
\underline{d}_\mu(x)=\liminf_{r\to 0}\frac{\log\mu( B(x,r))}{\log r}\ge  \frac{\log (b_1/2)+\log G_\ast(\alpha)}{d_2A^\ast(\alpha)}.
$$
  then by  Proposition \ref{fal}, we conclude that
$$
\dim_H \mathscr{C}\ge \frac{\log (b_1/2)+\log G_\ast(\alpha)}{d_2A^\ast(\alpha)}.
$$
This achieves the proof of Proposition \ref{prop-3.2}.
 \end{proof}
 
  \subsection{Proof of Theorem \ref{thm2}}
 Define
 $$
 C:=\max\{C_1,\tilde C (c_1,d_1),  2/b_1, (2/b_1)^2\}\ \ \text{ and }\ \ \ C':=2d_2.
 $$

 If $b_1\ge2$, then by Proposition \ref{prop-3.1} and \ref{prop-3.2},
 $$
 \dim_H\Sigma_{1,\alpha}\ge \dim_H\mathscr{C}\ge \frac{\log G_\ast(\alpha)}{d_2A^\ast(\alpha)}\ge \frac{1}{C'}\frac{\log G_\ast(\alpha)}{A^\ast(\alpha)}.
 $$

   If $b_1<2$, then we have $G_\ast(\alpha)\ge C\ge (2/b_1)^2$ and consequently
 $$
 \log G_\ast(\alpha)\ge -2\log (b_1/2).
 $$
Then by Proposition \ref{prop-3.1} and \ref{prop-3.2}, we have
 $$
 \dim_H\Sigma_{1,\alpha}\ge \dim_H\mathscr{C}\ge \frac{\log G_\ast(\alpha)+\log(b_1/2)}{d_2A^\ast(\alpha)}\ge \frac{\log G_\ast(\alpha)}{2d_2A^\ast(\alpha)}=\frac{1}{C'}\frac{\log G_\ast(\alpha)}{A^\ast(\alpha)}.
 $$ 
 
 Thus \eqref{Sigma-alpha} follows, and  \eqref{Sigma-alpha-n} is a direct consequence of \eqref{Sigma-alpha}.
  \qed
 
  \subsection{Proof of Theorem  \ref{thm1}}

% \noindent {\bf Proof of Theorem \ref{thm1}.}
%  Write
%$$
%\mathscr{F}:=\{\alpha\in (0,1)\setminus\Q:  \dimension(\Sigma_{1,\alpha})>0 \}\,.
%$$

At first, we show that $\mathscr{F}$ is dense in $(0,1)\setminus\Q$. Let $m=\widehat m=M.$ Let $C$ be the constant in Theorem \ref{thm2}. Define
$$
\mathscr{F}_M:= \{\alpha: a_i\le M, (1\le i\le M); a_i=C\, (i\ge M+1)\}.
$$
Then by Theorem \ref{thm2}, $\mathscr{F}_M\subset \mathscr{F}.$
%We also have  $\mathscr{F}_M \subset \{\alpha  \in (0,1)\setminus\Q| \beta(\alpha)=0 \}$, since
%$q_{i+1}\leq (C_M+1) q_{i}$ for all $i\ge M+1$.
On the other hand,  for any $ \alpha, \alpha'\in (0,1)\backslash \Q$, let  $n$ be  the
first index for which the continued fraction expansions of $\alpha$
and $\alpha'$ differ and define the distance of $\alpha$ and $\alpha'$ as
$$
d_{H}(\alpha, \alpha')=
\frac{1}{n+1}.
$$
 Endow  $  (0,1) \backslash\Q$ with this topology, then it is easy to see
 that this topology coincide with the usual topology induced from $\R$ and  $\bigcup_{M\ge1}\mathscr{F}_M$ is dense in $(0,1) \backslash\Q $. Hence, $\mathscr{F}$ is also dense in $(0,1) \backslash\Q$. 

Now we show that $\mathscr{F}$ has positive Hausdorff dimension. Indeed, by Theorem \ref{thm2}, we know that
$$
\widehat{\mathscr{F}}:=\{\alpha\in (0,1)\setminus\Q: C\le a_n\le 10\,C\, , n\ge1\}\subset \mathscr{F}.
$$
By Theorem 11 of \cite{Go}, $\widehat{\mathscr{F}}$ has positive Hausdorff dimension.
 \hfill $\Box$
 \appendix
 \section{The statement of Theorem 0.1 in \cite{HS2}}
 We give a translation from the french, correcting also a few typos and adding a few explanatory remarks.
 \begin{thm}
 Let $\hat m \in \mathbb N$  ($\hat m \geq 2$) and $M \geq 2$. There exists $\epsilon_1 >0$ and, for $\epsilon_0 \in (0,\epsilon_1)$,  a constant $C=C(\hat m, M,\epsilon_0) >0$ such that if $\alpha =[a_1,a_2,\dots,] $
  is irrational and satisfies  for some $m\leq \hat m$
 \begin{equation}\label{0.7}
 \begin{array}{ll}
 1 \leq  |a_j| \leq M& \mbox{ for } 0 < j \leq m\\
 |a_j| \geq C & \mbox{ for } j\geq m+1\,,
 \end{array}
 \end{equation}
 then $\Sigma_{1,\alpha}$ is contained in the union of $ q_m$ intervals $I_\ell (h)$ ($\ell=1,\cdots, q_m)$ in the form
 $[\gamma_\ell (h), \delta_\ell (h)]$ with
 \begin{equation}\label{0.8}
 \begin{array}{l}
   \gamma_\ell(h)\,,\, \delta_\ell (h) \in \Sigma_{1,\alpha}\,,\\
 \gamma_\ell < \delta_\ell \leq \gamma_{\ell +1} < \delta_{\ell +1}\,,\\
 \gamma_\ell(h) \geq \gamma_\ell  - C |h|\, \mbox{ and }   \delta_\ell (h) < \delta_\ell + C | h| \,, \\
\gamma_{\ell}(h) \geq \gamma_{\ell} + \frac{1}{C} \sqrt{h} \mbox{ if } \delta_{\ell-1} =\gamma_{\ell}\,,
 \end{array}
 \end{equation}
 where
 \begin{equation}
  \alpha^{(m)} =[a_1,\dots,a_m]= \frac{p_m}{q_m}\,,
  \end{equation}
 \begin{equation}
 h = 2\pi (\alpha - \alpha^{(m)}) \,,
 \end{equation}
 \begin{equation}\label{0.11}
 \cup_\ell [\gamma_\ell,\delta_\ell]  =\Sigma_{1,\alpha^{(m)}} \,,
 \end{equation}
 \begin{equation}\label{0.12}
 d(I_\ell (h), I_{\ell+1}(h))\geq \frac 1C \mbox{ if } \delta_\ell \neq \gamma_{\ell +1}
 \mbox{ and } \geq \frac 1 C \sqrt{|h|} \mbox{ if } \delta_\ell =\gamma_{\ell +1}\,.
 \end{equation}
  For each interval $I_\ell(h)$, $\Sigma_{1,\alpha}\cap I_\ell (h)$ can be described as living in a union of $N_{\ell,j}$ closed intervals $J_j^{(\ell)}$ (indexed by $j\in (- m_{\ell,j}, n_{\ell,j})$) of length $\neq 0$ with
 $\partial J_{j}^{(\ell)} \subset \Sigma_{1,\alpha}\,$, $J_{j+1}^{(\ell)}$ on the right of $J_j^{(\ell)}$ and
 \begin{equation}
 m_{\ell,j} \approx |a_{m+1}| \mbox{ and } n_{\ell,j} \approx |a_{m+1}|\,,
 \end{equation}
 \begin{equation}\label{0.13}
 \frac{1}{|a_{m+1}|} \lesssim d(J_j^{(\ell)}, J_{j+1}^{(\ell)} )  \lesssim \frac{1}{\sqrt{|a_{m+1}|}}\,,
 \end{equation}
 \begin{equation}
 J_0^{(\ell)} \mbox{ has length } 2 \epsilon_0 + \mathcal O (\frac{1}{|a_{m+1}|})\,.
 \end{equation}
 The other bands have size
 \begin{equation}\label{0.16}
 \exp \left( - C(j) |a_{m+1}| \right) \mbox{ with } C(j) \approx 1\,.
 \end{equation}
 For $j\neq 0$, if $\kappa_j^{(\ell)} $ is the affine function sending $J_j^{(\ell)}$ in $[-2,+2]$, then
 $$
 \kappa_j^{(\ell)}(J_{j}^{(\ell)}) \cap \Sigma_{1,\alpha} \subset \cup_k J_{j,k}^{(\ell)}\,,
 $$
 where the $J_{j,k}^{(\ell)}$ have analogous properties to the $J_j^{(\ell)} $ with $a_{m+1}$ replaced by $a_{m+2}$ and \eqref{0.13} can be improved in the form
 \begin{equation}
 d(J_{j,k}^{(\ell)}, J_{j,k+1}^{(\ell)}) \approx \frac{1}{|a_{m+2}|}\,.
 \end{equation}
 One can then iterate indefinitely.
   \end{thm}

%   {\clb Qi:  it seems to me \eqref{0.12} is unnatural,  if the rational band touches,  but  $d(I_\ell (h), I_{\ell+1}(h)) \geq \frac 1 C \sqrt{|h|}$, it means the irrational band still separates.  }
%   {\clr Bernard: in the even case, we  take, with $q= q^{(m)}$,   $\delta_\frac{q}{2}  (h) = \mathcal O (h)$ and  $\gamma_{\frac {q} {2}+1}(h)  = c \sqrt{h} + \mathcal O (h)$ . The spectrum is indeed contained in $\mathcal O (h)$-neighborhoods
%    of $\pm c \sqrt{nh}$ where $c \neq 0$ is computed using the Dirac approximation. See \cite{HS2}. \\
%    There is some arbitraryness in the choice. We could have taken $\delta_\frac{q}{2}  (h) = - c \sqrt{h}+  \mathcal O (h)$ and  $\gamma_{\frac {q} {2}+1}(h) =  \mathcal O (h)$. In the two choices I get a gap of order $\mathcal O (\sqrt{h})$. In the new writing of the theorem,  I have made the first choice. }

   Here in the statements $a \lesssim b$ means that $a/b \leq C$ where $C$ depends only on $C_0$ and $\epsilon_0$. The same is true when we use the notation  $\mathcal O$ or  $\approx\,$.

   \begin{rem} $\epsilon_0$ corresponds  with the exclusion in each interval and at each step of the renormalization of a small interval of size $\approx 2\epsilon_0$
    for which another analysis has to be done and which was the object of \cite{HS3} (see also \cite{HeKe}). This corresponds  to the energy $0$ for the map $(x,\xi) \mapsto 2( \cos x + \cos \xi$). This refined analysis is not needed here.
    \end{rem}

    \begin{rem}
    The possibility of having $\delta_\ell =\gamma_{\ell +1}$ is due to the occurence of touching bands. Van Mouche \cite{VM}  has proven that it occurs only
     when $q_m$ is even and  for $\ell = \frac{ q_m} {2}$.
     These two touching bands lead to the lower bound \eqref{0.12} and  the weaker estimate in \eqref{0.13}.
     \end{rem}

\section{Acknowledgements}
Q.-H. Liu was supported by NSFC grant (11571030).
Y.-H. Qu was supported by NSFC grant (11431007 and 11790273). 
Q. Zhou was partially supported by NSFC grant (11671192), \textquotedblleft Deng Feng Scholar Program B\textquotedblright of Nanjing University,  Specially-appointed professor programme of Jiangsu province.

 %%%%%%%%%%%%%%%%%%%%%%%%%%%%%%%%%%%%%%


\begin{thebibliography}{30}
%\bibitem{AA80} S. Aubry and G. Andr\'{e},
%Analyticity breaking and Anderson localization in incommensurate
%lattices. In: Group Theoretical Methods in Physics (Proc. Eighth
%Internat. Colloq. Kiryat Anavim, 1979), Hilger, Bristol, pp. 133-164
%(1980).
%
%
%\bibitem{Aab}
%A.  Avila, The absolutely continuous spectrum of the almost Mathieu operator,  arXiv:0810.2965.


\bibitem{AK06}
A. Avila and R. Krikorian,
Reducibility or non-uniform hyperbolicity for quasi-periodic
Schr\"{o}dinger cocycles. Ann. Math. \textbf{164}, 911-940  (2006).

\bibitem{AJ05}
 A. Avila  and  S. Jitomirskaya,
The Ten Martini Problem,
 Ann. Math. \textbf{170},  303-342  (2009).
 
 
\bibitem{AJ08}
 A. Avila  and  S. Jitomirskaya, 
Almost localization and almost reducibility,
 J. Eur. Math. Soc, \textbf{12}, 93-131 (2010).
 

 \bibitem{ALSZ}
A. Avila, Y. Last, M. Shamis, and Q. Zhou,
\newblock On the abominable properties of the Almost Mathieu operator with well approximated frequencies,
\newblock  in preparation.

%
%\bibitem{AYZ1} A. Avila, J. You and Z. Zhou, Sharp Phase transitions  for the almost Mathieu operator,  {\it Duke Math. J.},  \textbf{166}, 2697-2718 (2017).


\bibitem{AYZ2} A. Avila, J. You and Z. Zhou, The Dry Ten Martini Problem in the non-critical case, {\it preprint}.

\bibitem{AMS}J.Avron, P. van Mouche, B. Simon, On the measure of the spectrum for the almost Mathieu operator. Commun. Math. Phys. 132, 103-118 (1990)

\bibitem{AOS} J. E. Avron, D. Osadchy and  R. Seiler,   A topological look at the quantum Hall effect. Physics today.  38-42. (2003).
\bibitem  {bel} J. Bellissard,
\newblock Le papillon de Hofstadter.
\newblock  Ast\'erisque {\bf 206} (1992) 7--39.

\bibitem{conjhalf2} J. Bell and R. B. Stinchcombe, Hierarchical band clustering and
fractal spectra in incommensurate systems. J.\ Phys.\ \textbf{A}
$\mathbf{20}$,\ L739--L744 (1987).

 \bibitem{Fal} K. Falconer,
\newblock  Techniques in fractal geometry.
\newblock John Wiley $\&$ Sons, Ltd., Chichester, 1997.

\bibitem{conjhalf3} T. Geisel, R. Ketzmerick and G.Petshel, New class of
level statistics in quantum systems with unbounded diffusion.
Phys.\ Rev.\ Lett.\ $\mathbf{66}$,\ 1651--1654 (1991).

\bibitem{Go} I.J. Good,
The fractional dimensional theory of continued fractions.
Proc. Cambridge Philos. Soc. 37, (1941). 199--228.


\bibitem{Ha} P.G. Harper,  Single band motion of conduction electrons in a uniform magnetic
field, Proc. Phys. Soc. London A. 68, 874--892 (1955).

 \bibitem{HeKe} B. Helffer and P. Kerdelhu\'e.
\newblock On the Total Bandwidth for the Rational Harper's Equation.
\newblock Commun. Math. Phys. 173, 335--356 (1995).


 \bibitem{HS}  B. Helffer and J. Sj\"ostrand, Analyse semi-classique pour l'\'equation de Harper (avec application  \`a   l'\'equation de Schr\"odinger avec champ magn\'etique).
  M\'em. Soc. Math. France,  No. 34,  1--113  (1988).

 \bibitem{HS2}   B. Helffer and J. Sj\"ostrand, Analyse semi-classique pour l'\'equation de Harper II: Comportement semi-classique  pr\`es  d'un rationnel. M\'em. Soc. Math. France, No. 40, 1--139. (1990).

{ \bibitem{HS3}  B. Helffer and J. Sj\"ostrand, Semi-classical analysis for the Harper's equation III: Cantor structure of the spectrum.
\newblock M\'em. Soc. Math. France, No. 39. 1--124 (1989).}


%\bibitem{Ji95}
%S. Jitomirskaya,
%Almost Everything About the Almost Mathieu Operator, II.
%Proceedings of XI International Congress of Mathematical
%Physics,Int. Press, (1995), 373-382.
%
%
%\bibitem{J99} S. Jitomirskaya,  Metal-insulator transition for the
%almost Mathieu operator.  Ann. of Math. (2)  150  (1999),  no. 3,
%1159--1175.
%
%
%
%\bibitem{J07}
% S. Jitomirskaya,
%Ergodic Schr\"{o}dinger operator (on one foot), Proceedings of
%symposia in pure mathematics. Volumn 76.2, 613-647 (2007)
%
%\bibitem{JL} S. Jitomirskaya and W. Liu, Universal hierarchical structure of quasi-periodic eigenfunctions, arXiv: 1609.08664.


\bibitem{JZ} S. Jitomirskaya and S. Zhang, Quantitative continuity of singular continuous spectral measures and arithmetic criteria for quasiperiodic Schr\" odinger operators. ArXiv preprint in arXiv:1510.07086 (2015).


 \bibitem{JM82}
R. Johnson and J. Moser,
The rotation number for almost periodic potentials. Comm. Math.
Phys, 84, 403-438 (1982).


\bibitem{k}I. Krasovsky. Central spectral gaps of the Almost Mathieu Operator.  Commun. Math. Phys. 351, 419-439 (2017).

\bibitem{last3} Y. Last,  Zero measure spectrum for the almost
Mathieu Operator. Commun.\ Math.\ Phys.\ $\mathbf{164}$,\ 421--432
(1994).

\bibitem{LS} Y. Last and M. Shamis, Zero Hausdorff Dimension Spectrum for the Almost Mathieu Operator. Commun.\ Math.\ Phys. 348,  729--750 (2016).



\bibitem{OA} D. Osadchy and J. E. Avron,  Hofstadter butterfly as quantum phase
diagram.  J. Math Phys. 42, 5665-5671 (2001).

\bibitem{Pe} R. Peierls, Zur Theorie des Diamagnetismus von Leitungselektronen.  Z. Phys. 80, 763-791 (1933).

\bibitem{R} A. Rauh,  Degeneracy of Landau levels in crystals, Phys. Status Solidi B 65, 131-135 (1974).

\bibitem{Ro}H. L. Royden,  Real analysis. Third edition. Macmillan Publishing Company, New York, 1988.

\bibitem{conjhalf1} C. Tang and M. Kohmoto, Global scaling properties of the
spectrum for a quasiperiodic Schr\"{o}dinger equation. Phys.\ Rev.\
\textbf{B} $\mathbf{34}$,\ 2041--2044 (1986).

\bibitem{TKNN} D.J. Thouless,  M. Kohmoto,   M.P.  Nightingale, and   M. Den Nijs,  Quantized Hall conductance in a two dimensional periodic potential. Phys. Rev. Lett. 49, 405--408 (1982).

{
\bibitem{VM} P. Van Mouche.
\newblock The coexistence problem for the discrete Mathieu operator.
\newblock Comm. Math. Phys. 122 (1989), no. 1, 23--33.

}
\bibitem{wilkinson_austin} M. Wilkinson and E.J. Austin, Spectral
dimension and dynamics for Harper's equation. Phys.\ Rev.\
\textbf{B} $\mathbf{50}$,\ 1420--1430 (1994).
\end{thebibliography}
\end{document}